\documentclass[12pt, a4paper]{extarticle}
\usepackage[margin=2cm]{geometry}
\usepackage[affil-it]{authblk}

\usepackage[english]{babel}
\usepackage{graphicx} 
\usepackage{wrapfig}
\usepackage{amsmath,amsthm,amssymb,scrextend}
\usepackage{booktabs}            
\usepackage{multirow}            
\usepackage{amsfonts}            
\usepackage{multicol}
\usepackage[shortlabels]{enumitem}
\usepackage{cutwin}
\usepackage{float}
\usepackage{csquotes}
\usepackage{enumitem}
\usepackage{tikz}
\usepackage{hyperref}

\def \rR {\mathbb{R}}
\def \bf#1 {\textbf{#1 }}
\def \sumt {\sum\limits}
\providecommand{\keywords}[1]
{
  \small	
  \textbf{\textit{Keywords:}} #1
}

\renewenvironment{proof}{\begin{addmargin}[1em]{0em}\begin{newproof}}{\end{newproof}\end{addmargin}\qed}
\newtheorem{thm}{Theorem}
\newtheorem{lm}{Lemma}
\newtheorem{cor}{Corollary}
\newtheorem*{lm*}{Lemma}
\newtheorem{defin}{Definition}
\newtheorem*{defin*}{Definition}

\def \sumt {\sum\limits}
\DeclareMathOperator{\diam}{diam}
\DeclareMathOperator{\dist}{dist}

\DeclareMathOperator{\Clo}{Clo}
\DeclareMathOperator{\Rad}{Rad}
\DeclareMathOperator{\Str}{Str}

\begin{document}

\title{Relations between average shortest path length and another centralities in graphs}
\author{Mikhail Tuzhilin%
  \thanks{Affiliation: Moscow State University, Electronic address: \texttt{mtu93@mail.ru};}
}
\date{}
\maketitle

\begin{abstract}
Relations between average shortest path length and radiality, closeness, stress centralities and average clustering coefficient were obtained for simple connected graphs.
\end{abstract}



\section{Introduction.}
Centrality measures are important characteristics in  network science which contain intrinsic hidden information about networks. Centrality is a local (with relation to a vertex) or global (with relation to a whole graph) real functions on graphs. There exist many centralities~\cite{Borgatti}-~\cite{Lee}: local efficiency, radiality, closeness, betweenness, stress centralities, etc. One of the most important centralities for ``real'' networks (or small-world networks) are average clustering coefficient and average shortest path length~\cite{Watts}. The network is called small-world if it has small average shortest path length with relation to the size and big average clustering coefficient. This measures are separately calculated to prove that current network is small-world. In this article we proved that for geodetic graphs (or graphs with odd length cycles) with additional condition there is relation between average shortest path length and average clustering coefficient.

\section{Main definitions.}

All subsequent definitions are given for a simple, connected, undirected graph $G$. Let's denote by
\begin{itemize}
    \item $V(G)$ the set of vertices,\;\; $E(G)$ the set of edges and\;\; $A = \{a_{ij}\}$ adjacency matrix of graph $G$.
    \item Neighborhood $N(v)$ --- the induced subgraph in $G$ on vertices which adjacent to the vertex $v$,
    \item $N'(v)$ induced subgraph in $G$ on vertices $ = V\big(N(v)\big)\bigcup \{v\}$,
    \item $\bar{f}(x_1, x_2, ... , x_k)$, where $f$ is any function $V\times V \times ... \times V\rightarrow\rR$, the restriction of this function on $N'(v)$ (for example $\bar L(x,y)$ will be the average shortest path between $x$ and $y$ with restriction to subgraph $N'(v)$),
    \item $d_i = deg(v_i)$,
    \item $n = |V(G)|,\;\; m = |E(G)|$,
    \item $X(i) = X(v_i)$ for any $X$ --- set or function corresponding to vertex $v_i$,
\end{itemize}

Let's give definitions of centralities:
\begin{enumerate}
    \item \bf{Diameter} $diam(G) = \max_{s,t\in V(G)} dist(s,t)$.
    \item \bf{Average shortest path length} $L(G) = \frac 1 {n(n-1)} \sumt_{s\neq t} dist(s,t)$.
    \item \bf{Local cluster coefficient} 

    $c_i = c(i) = \frac {\text{number of edges in }N(i)} {\text{maximum possible number of edges in }N(i)}= \frac{2 |E(N(i))|} {d_i(d_i-1)}$. 
    \item \bf{Average clustering coefficient} 
    
    $C_{WS}(G) = \frac 1 {n} \sumt_{i\in V(G)} c_i = \frac 1 {n} \sumt_{i\in V(G)}  \frac{2 |E(N(i)))|} {d_i(d_i-1)} = \frac 1 n \sumt_{i\in V(G)} \frac {\sumt_{j, k\in V(G)} a_{ij} a_{jk} a_{ki}} {d_i (d_i-1)}$.
    \item \bf{Global clustering coefficient} 
    
    $C(G) = \frac {\text{number of closed triplets in $G$}} {\text{number of all triplets in $G$}} = \frac {\sumt_{i, j, k\in V(G)} a_{ij} a_{jk} a_{ki}} {\sumt_{i\in V(G)} d_i (d_i-1)}$.
    \item \bf{Closeness centrality} $Clo(v) = \frac {n-1} {\sumt_{t\in V(G)} dist(v, t)}$. 
    \item \bf{Radiality} $Rad(v) = \frac {\sumt_{t\in V(G), t\neq v} (diam(G)+1-dist(v, t))} {n-1}$.
    \item \bf{Stress} $Str(i) = \sumt_{s,t\in V(G),\; s\neq t\neq i} \sigma_{st}(i)$, where $\sigma_{st}(i)$ is the total number of shortest paths from $s$ to $t$ which contains vertex $i$.
\end{enumerate}

Note that all centralities are non-negative and $c_i, C_{WS}, \Clo(v)$ are less or equal 1.

Also let's give the definition of geodetic graph:
\begin{defin}
    If there exists the unique shortest path between any two vertices in $G$ then the graph $G$ is called \bf{geodetic}.
\end{defin} 
This definition is equivalent to the condition then there is no even cycles in a graph.

\section{Main results.}

Let's consider a induced subgraph $G'\subset G$. In general $G'$ can be not connected graph. In this case let's define the average shortest path length for vertices of $G'$ with relation to the distance $\dist$ in the ambient graph $G$. Let's call $L\big(N(i)\big)$ the local average shortest path length for the vertex $i$. 

Let's start with simple relations: let's proof the relation between local shortest path length and local clustering coefficient:

\begin{lm}\label{lm1}\label{lm1}
    $$L(N(i)) = 2-c_i.$$
\end{lm}
\begin{proof}
 $$
L(N(i)) = \frac 1 {d_i(d_i-1)} \sumt_{s,t\in N(i), s\neq t} dist(s,t) = \frac 1 {d_i(d_i-1)} \sumt_{(s,t)\in E(N(i))} dist(s,t)+ \sumt_{s,t\in N(i), (s,t)\notin E(N(i))} dist(s,t) = 
$$
$$
 = \frac 1 {d_i(d_i-1)} (2 |E(N(i))| + \sumt_{(s,i), (i, t)\in E(G), (s,t) \notin E(G)} dist(s,t)) = 
$$ 
$$ 
= \frac 1 {d_i(d_i-1)} (2 |E(N(i))|+2 (d_i(d_i-1) - 2 | E(N(i))|)) = 2 - c_i.
$$
Note that shortest paths for vertices in $N(i)$ are defined corresponding to whole graph $G$.
\end{proof}

Averaging by $i$ we obtain simple corollary about the relation between local shortest path length and average clustering coefficient.

\begin{cor}
    $$C_{WS}(G) = 2- \frac 1 n \sumt_{i\in V(G)}L\big(N(i)\big).$$
\end{cor}

Let's proof the relation between shortest path length in subgraph $N'(i)$ and shortest path length in $N(i)$.

\begin{lm}
    $$L\big(N'(i)\big) =  \frac {(d_i-1)L\big(N(i)\big)+2} {d_i+1}.$$
\end{lm}
\begin{proof}
    By definition
    $$
        L\big(N'(i)\big) = \frac 1 {(d_i+1)d_i} \sumt_{s,t\in V(N'(i)),\;s\neq t} \dist(s,t) = \frac {d_i-1} {d_i+1} L\big(N(i)\big) + \frac 2 {d_i+1}.
    $$
\end{proof}

Let's proof the relation between shortest path length in a induced subgraph and shortest path length in ambient graph if induced subgraph is obtained from ambient graph by deleting one vertex.

\begin{thm}
    Let a graph $G$ is obtained from a connected simple graph $G'$ by deleting one vertex and $|V(G)| = n$. Then
    $$L(G')\geq \frac n {n+1} L(G),$$
    where the average shortest path length $L(G)$ is defined in the ambient graph $G'$, if $G$ is not connected.
\end{thm}
\begin{proof}
    Let's define the deleted vertex by $v$. By the triangle inequality $\forall s, t\in V(G): \dist(s,v)+\dist(v,t)\geq \dist(s,t)$, where the equality holds if, there are no paths from $s$ to $t$ in $G$. Therefore,
    $$\sumt_{s,t\in V(G), s\neq t}\big(\dist(s,v)+\dist(v,t)\big)\geq \sumt_{s,t\in V(G), s\neq t}\dist(s,t)$$
    $$\frac {2(n-1)} {n(n-1)}\sumt_{t\in V(G)}\dist(v,t)\geq\frac 1 {n(n-1)} \sumt_{s,t\in V(G), s\neq t}\dist(s,t)$$
    $$\frac {2} {n}\sumt_{t\in V(G)}\dist(v,t)\geq L(G)$$
    Then,
    $$L(G') = \frac 1 {(n+1)n} \sumt_{s,t\in V(G'),s\neq t}\dist(s,t) = \frac 1 {(n+1)n} \Big(2 \sumt_{t\in V(G)}\dist(v,t)+\sumt_{s,t\in V(G),s\neq t}\dist(s,t)\Big) =$$
    $$=\frac 1 {n+1} \frac 2 {n} \sumt_{t\in V(G)}\dist(v,t) + \frac {n-1} {n+1} L(G)\geq \frac n {n+1} L(G)$$
    Note that if $G$ consists of $n$ isolated vertices then the equality holds.
\end{proof}

We obtain corollaries

\begin{cor}
     Let a graph $G$ is obtained from a connected simple graph $G' \subset H$ by deleting one vertex, $|V(G)| = n$ and $H$ is connected and simple. Then
    $$L(G')\geq \frac n {n+1} L(G),$$
    where the average shortest path lengths $L(G)$ and $L(G')$ are defined in the ambient graph $H$, if corresponded graphs are not connected.
\end{cor}
\begin{proof}
The proof is the same as for the previous theorem and also if $G$ consists of $n$ isolated vertices then the equality holds.
\end{proof}

\begin{cor}\label{cor1}
    $$L\big(N'(i)\big) \geq  \frac {d_i} {d_i+1}L\big(N(i)\big).$$
\end{cor}

Let's proof the relation between shortest path length in a induced subgraph and shortest path length in ambient graph.

\begin{thm}
    Let's $G'\subset G$ be induced subgraph and $|V(G)| = n, |V(G')| = n+k$. Then
    $$L(G')\geq \frac n {n+k} L(G),$$
    where the average shortest path length $L(G)$ is defined in the ambient graph $G'$, if $G$ is not connected.
\end{thm}
\begin{proof}
    Let's construct the graph $G'$ from $G$ by adding sequentially $k$ vertices and corresponded edges. Let's first sequentially add vertices adjacent to vertices of the graph $G$. Adding these vertices one by one we obtain a sequence of graphs $G_1, G_2, ... , G_p, $ where $|V(G_i)| = n+i$. Further, let's add is the same way vertices adjacent to vertices of the graph $G_p$ and so on. In the end we obtain the graph $G'$. By previous corollary
    $$\hspace{-15pt}L(G')\geq \frac {n+k-1} {n+k} L(G_{k-1})\geq \frac {n+k-1} {n+k} \frac {n-k-2} {n-k-1} L(G_{k-2}) = \frac {n-k-2} {n+k} L(G_{k-2})\geq \cdots \geq \frac n {n+k} L(G).$$
\end{proof}

Let's proof the relation between average closeness centrality and average shortest path length:

\begin{lm}
$$
L(G)\geq \frac n {\sumt_{v\in V(G)} \Clo(v)}.
$$
\end{lm}
\begin{proof}
By the inequality of harmonic mean and arithmetic mean
$$
\frac 1 {n} \sumt_{v\in V(G)} Clo(v) = \frac 1 n \sumt_{v\in V(G)} \frac {n-1} {\sumt_{t\in V(G)} dist(v, t)}\geq \frac {n(n-1)} {\sumt_{v, t\in V(G)} dist(v, t)} = \frac 1 {L(G)}.
$$
Note that an equality holds when all average shortest path lengths from any vertex to all remaining vertices are equal.
\end{proof}

Let's proof the relation between average shortest path length and average radiality:

\begin{lm}\label{lm3}
$$
 L(G) = \diam(G)+1 -\frac 1 {n} \sumt_{v\in V(G)} {\Rad}(v).
$$
\end{lm}
\begin{proof}
The proof holds from definition
$$
\frac 1 {n} \sumt_{v\in V(G)} {Rad}(v) = \frac 1 {n} \sumt_{v\in V(G)} \frac {(n-1) (diam(G)+1)- \sumt_{t\in V(G),\; t\neq v} dist(v,t))} {n-1} = diam(G)+1 -L(G).
$$
\end{proof}

Now let's prove theorem about a relation between average clustering coefficient and radiality using previous theorem.
\begin{thm}
$$
\frac {C_{WS}(G)} 2 \geq \frac 1 n \sumt_{i\in V(G)} \frac {\sumt_{v\in N(i)} \overline{\Rad}(v)} {d_i} -2 + \frac {\#\{N(i) \text{ which are complete graphs}\}} {n}.
$$
\end{thm}
\begin{proof}
By lemmas~\ref{cor1}: $L(N'(i))\geq \frac {d_i} {d_i+1}L(N(i))\geq\frac 1 2 L(N(i)) = 1-\frac 1 2 c_i$. Therefore,
$$
\frac 1 {d_i} \sumt_{v\in N(i)} \overline{\Rad}(v) =  \diam(N'(i))+1-L(N'(i)) \leq \diam (N'(i))+\frac 1 2 c_i = \frac 1 2 c_i+2-\chi_{K_{d_i}}(N(i)),
$$
where 
$
\chi_{K_{d_i}}(N(i)) =
\begin{cases}
	1 & \text{if $N(i) = K_{d_i}$} \\
	0 & \text{otherwise}
\end{cases}
$. Averaging by $i$ ends the proof.
\end{proof}

Let's prove a theorem about a connection between the average stress centrality and average shortest path length for geodetic graphs.

\begin{thm}\label{thm1}
    If $G$ is geodetic, then
    $$
        L(G) = 1+\frac {1} {n(n-1)} \sumt_{i\in V(G)}\Str(i).
    $$
\end{thm}
\begin{proof}
    Let's define 
    $$\chi_{st}(i) = \begin{cases}
        1 & \text{if $i\neq s\neq t$ is the vertex of the shortest path between $s$ and $t$,}\\
        0 & \text{otherwise.}
    \end{cases}$$
    
    Thus, $\Str(i) = \sumt_{s,t\in V(G)} \chi_{st}(i)$. If there exists the unique shortest path between any two vertices in $G$, then $\dist(s,t) = \sumt_{i\in V(G)}  \chi_{st}(i)+1$ (otherwise, $\dist(s,t) \leq \sumt_{i\in V(G)}  \chi_{st}(i)+1$).  Therefore, for any $i$
    $$
        \sumt_{s,t\in V(G)} \dist(s,t) = 2|E|+\sumt_{s,t\in V(G),\,\dist(s,t)\geq 2} \dist(s,t) = 2|E|+\sumt_{s,t\in V(G),\,\dist(s,t)\geq 2} \Big(\sumt_{i\in V(G)} \chi_{st}(i)+1\Big) =
    $$
    $$
        = 2|E|+\sumt_{i\in V(G)}\Str(i)+n(n-1)-2|E|=\sumt_{i\in V(G)}\Str(i)+n(n-1).
    $$
\end{proof}

\begin{cor}
    For any simple connected $G$
    $$
        L(G) \leq 1+\frac {1} {n(n-1)} \sumt_{i\in V(G)}\Str(i).
    $$
\end{cor}

First let's proof theorem about a connection between average clustering coefficient and stress centrality and we will use it for a connection between the average shortest path length and the average local clustering coefficient further.

\begin{thm}\label{thm2}
If there is no pendant vertices in graph, then
$$
C_{WS}(G)\geq 1- \frac 1 n \sumt_{i\in V(G)} \frac {Str(i)} {d_i(d_i-1)}.
$$
\end{thm}
\begin{proof}
Note that $\forall j,k\in N(i): (j,k)\notin E(N(i))$ the shortest path between $j$ and $k$ is $j\rightarrow i\rightarrow k$. Therefore,
$$
Str(i)\geq 2( \frac {d_i(d_i-1)} 2-|E(N(i))|),
$$
$$
\frac 1 {d_i(d_i-1)} Str(i)\geq 1-c_i,
$$
Averaging by $i$
$$
C_{WS}(G) \geq \frac 1 n \sumt_{i\in V(G)} (1- \frac {Str(i)} {d_i(d_i-1)}).
$$
Note that for $diam(G) = 2$ holds an equality.
\end{proof}

\begin{cor}
    In the general case
    $$
    C_{WS}(G)\geq 1- \frac 1 n \sumt_{i\in V(G)} \frac {Str(i)} {d_i(d_i-1)}-\frac {\text{number of pendant vertices}} n.
    $$
\end{cor}
\begin{proof}
    Let's for pendant vertex define $\frac {Str(i)} {d_i(d_i-1)}$ as 0, then in the inequality $\frac 1 {d_i(d_i-1)} Str(i)\geq 1-c_i$ in the right side will be 1 and in the left side 0, thus if we add in the left side 1 for every pendant vertex, will be right equality. 
\end{proof}

Now let's prove a theorem about a connection between the average shortest path length and the average local clustering coefficient for geodetic graphs.
\begin{thm}\label{thm6}
    If $G$ is geodetic and $\forall i,j\in V(G)$ hold, if $d_i\leq d_j $ then $ \Str(i)\leq \Str(j)$, then
    $$
        1-C_{WS}(G) \leq\frac 1 n \sumt_{i\in V(G)} \frac {\big(L(G)-1\big)(n-1)} {d_i(d_i-1)}+\frac {\text{number of pendant vertices}} n.
    $$
\end{thm}
\begin{proof}
Let's re-numerate vertices such that $\forall i\leq j:d_i\leq d_j$. Then for $i\leq j$ hold $\Str(i)\leq \Str(j)$ and $d_i(d_i-1)\leq d_j(d_j-1)$. By theorems~\ref{thm1} and~\ref{thm2} for the case if there is no pendant vertices in graph
    $$\hspace{-150pt}
        1-C_{WS}(G)\leq \frac 1 n \sumt_{i\in V(G)} \frac {Str(i)} {d_i(d_i-1)}\stackrel{\text{Chebyshev's sum inequality}}{\leq} 
    $$
    $$\hspace{70pt}
        \leq\Big(\frac 1 n \sumt_{i\in V(G)} {Str(i)}\Big) \Big(\frac 1 n \sumt_{i\in V(G)} \frac {1} {d_i(d_i-1)} \Big)= \frac 1 n\sumt_{i\in V(G)} \frac {\big(L(G)-1\big)(n-1)} {d_i(d_i-1)}.
    $$
    For pendant vertices we should add $\frac {\text{number of pendant vertices}} n$ in the right side. Note that if there exist two vertices $i,j\in V(G)$ such that $d_i < d_j $ and $ \Str(i)< \Str(j)$ then the inequality in this theorem will be strict.
\end{proof}

Let's consider a star $G$ with $V(G) = n+1$ vertices. The star is geodetic graph. The central vertex has degree $n$, local clustering coefficient $c_i = 0$ and stress centrality $\Str(i) = n(n-1)$. Other vertices are pendant ($d_i = 1,\ c_i = 0,\ \Str(i) = 0$). Thus for this graph holds theorem~\ref{thm6}. 
$$
L(G) = \frac {n(2n-1)+n} {(n+1)n} = \frac {2n} {n+1},\quad C_{WS} = 0.
$$
$$
1-C_{WS}(G) = 1 = \frac {\frac {n-1} {n+1}\, n} {n(n-1)}+\frac n {n+1} =  \frac 1 {n+1} \sumt_{i\in V(G)} \frac {\big(L(G)-1\big)n} {d_i(d_i-1)}+\frac {\text{number of pendant vertices}} {n+1}.
$$
Thus, for this example holds equality.


\begin{thebibliography}{99}

\bibitem{Borgatti}
Borgatti S. P., Everett M. G. A graph-theoretic perspective on centrality //Social networks. 2006. \bf{28}. №~4. 466--484.

\bibitem{Kiss}
Kiss C., Bichler M. Identification of influencers—measuring influence in customer networks //Decision Support Systems. 2008. \bf{46}. №~1. 233--253.

\bibitem{Lee} 
Lee S. H. M., Cotte J., Noseworthy T. J. The role of network centrality in the flow of consumer influence //Journal of Consumer Psychology. 2010. \bf{20}. №~1. 66--77.

\bibitem{Watts} 
Watts D. J., Strogatz S. H. Collective dynamics of ‘small-world’networks //nature. 1998. \bf{393}. №~6684. 440--442.


\end{thebibliography}
\end{document}